\documentclass[10pt,a4paper]{article}
\title{\bf Optimal Designs for Second-Order Interactions in Paired Comparison Experiments with Binary Attributes}
\author{Eric Nyarko\footnote{corresponding author:\texttt{eric.nyarko@ovgu.de}}, Rainer Schwabe \footnote{E-mail: \texttt{rainer.schwabe@ovgu.de}}\\
University of Magdeburg, \\
Institute for Mathematical Stochastics,\\ PF 4120, D-39016 Magdeburg, Germany}
\usepackage{amsmath}
\usepackage{booktabs}
\usepackage{float}
\usepackage{diagbox}
\usepackage{slashbox}
\usepackage{appendix}
\usepackage{tikz}
\usepackage{natbib}
\usepackage{longtable,pdflscape,caption,hyperref}
\usepackage{amsmath,mathtools}
\usepackage{amsthm}

\date{}
\begin{document}
\maketitle
\begin{abstract}
In paired comparison experiments respondents usually evaluate pairs of competing options. For this situation we introduce an appropriate model and derive optimal designs in the presence of second-order interactions when all attributes are dichotomous.
\end{abstract}
\section{Introduction}
Paired comparisons are closely related to experiments with choice sets of size two. For this situation optimal designs are usually derived under the indifference assumption of equal choice probabilities where the information matrix of a paired comparison experiment in a linear paired comparison model is equivalent to the information matrix of a discrete choice experiment in a multinomial  logit model. Such experiments have received considerable attention during the last few years in many fields of applications like psychology, health economics, transport and marketing for learning consumer preferences towards new products or services. Typical with paired comparisons, judges evaluate pairs of competing alternatives in a hypothetical setting which are generated by an experimental design, and are characterized by a number of attribute levels. A comprehensive introduction to the general area can be found in the monographs of \citet{louviere2000stated} and \citet{train2003discrete} \citep[see also][]{grossmann2015handbook}.
\par
The aim of this paper is to introduce an appropriate model and derive optimal designs in the presence of interactions. In this paper we treat the case when the components of the alternatives are characterized by two-level attributes. This factorial $2^{K}$ setting has been investigated by both \citet{van1987optimal1, van1987optimal} and \citet{street2004optimal} in the case of full profiles in a main effects and first-order (two factor) interactions situation, and by \cite{schwabe2003optimal} for partial profiles. Corresponding results have been presented by \citet{grasshoff2003optimal} in the general level case in a first-order interactions setup for both full and partial profiles. Here, we treat the case of second-order (three factor) interactions and provide detailed proofs. 
\par
The remainder of the paper is organized as follows. In Section 2 a general model is introduced for paired comparisons which is related to a block (choice set) of size two. The second-order interactions model for full profiles is presented in Section 3. Optimal designs are characterized in Section 4 and the final Section 5 offers some conclusions. The proof of the major results is deferred to the Appendix.

\section{General setting}
The outcome of any experimental situation depends on some factors $K$ of influence which are normally referred to as attributes in the paired comparison experiment literature. This dependence is best described by a functional relationship $\textbf{f}$ which quantifies the effect of the alternative $\textbf{i}$ of the attributes of influence. Hence, we formalize the experimental situation by a general linear model
\begin{equation}\label{eq:1}
\begin{split}
\tilde{Y}_{na}(\textbf{i})&= \mu_{n}+\textbf{f}(\textbf{i})^{\top} \boldsymbol{\beta} + \tilde{\varepsilon}_{na},
\end{split}
\end{equation}
for the value (utility) $\tilde{Y}_{na}(\textbf{i})$ of a single alternative $\textbf{i}$ within a pair of alternatives ($a=1,2$) subject to a random error $\tilde{\varepsilon}_{na}$. The index $n$ denotes the $n$th presentation in which $\textbf{i}$ is chosen from a set $\boldsymbol{\mathcal{I}}$ of possible realizations for the alternative. In general, each alternative is characterized by a number of distinct attributes (components) of influence such that $\textbf{i}= (i_{1},\dots,i_{K})$ for $k=1,\dots,K$. In this setting $\textbf{f}= (f_{1},\dots,f_{p})^{\top}$ is a vector of known regression functions which describe the form of the functional relationship between the alternative $\textbf{i}$ and the corresponding mean response $E(\tilde{Y}_{na}(\textbf{i}))=\mathbf{f}(\mathbf{i})^{\top}\boldsymbol{\beta}$, and  $\boldsymbol{\beta}= (\beta_{1},\dots,\beta_{p})^{\top}$ is the unknown parameter vector of interest. Usually in order to make statistical inference on the unknown parameters more than one observation is presented  to get rid of the influence of the presentation effect $\mu_{n}$ due to a variety of unobservable influences. 
\par
However, unlike standard design problems, in paired comparison experiments the utilities for the alternatives are usually not directly observed. Only observations $Y_n(\textbf{i},\textbf{j})=\tilde{Y}_{n1}(\textbf{i})-\tilde{Y}_{n2}(\textbf{j})$ of the amount of preference are available for comparing pairs $(\mathbf{i},\mathbf{j})$ of alternatives $\textbf{i}$ and $\textbf{j}$ which are chosen from the design region $\mathcal{X}=\boldsymbol{\mathcal{I}}\times \boldsymbol{\mathcal{I}}$. In that case the utilities for the alternatives are properly described by the linear paired comparison model 
\begin{equation} \label{eq:2}
\begin{split}
Y_n(\textbf{i, j})=(\textbf{f}(\textbf{i})-\textbf{f}(\textbf{j}))^{\top}\boldsymbol\beta+\varepsilon_{n}, \\
\end{split}
\end{equation}
where $\textbf{f}(\textbf{i})-\textbf{f}(\textbf{j})$ is the derived regression function and the random errors $\varepsilon_{n}(\textbf{i}, \textbf{j})=\tilde{\varepsilon}_{n1}(\textbf{i})-\tilde{\varepsilon}_{n2}(\textbf{j})$ associated with the different pairs $(\textbf{i}, \textbf{j})$ are assumed to be uncorrelated with constant variance. Here, the pair effects $\mu_{n}$ become immaterial. Moreover, we note that the linear difference model considered here can be realized as a linearization of the binary response model by \citet{bradley1952rank} under the assumption $\boldsymbol{\beta}=\mathbf{0}$ \citep[see e.g.][]{grossmann2002advances}. Specifically, under this indifference assumption of equal choice probabilities, the Bradley-Terry type choice experiments in which the probability of choosing $\textbf{i}$ from the pair $(\textbf{i}, \textbf{j})$ given by $\exp[\textbf{f}(\textbf{i})^{\top}\boldsymbol{\beta}]/(\exp[\textbf{f}(\textbf{i})^{\top}\boldsymbol{\beta}]+\exp[\textbf{f}(\textbf{j})^{\top}\boldsymbol{\beta}])$ as in the work of \citet{street2007construction} can be derived by considering the linear paired comparison model. 
\par
The quality of a design is measured by its information matrix
\begin{equation}\label{eq:3}
\textbf{M}((\textbf{i}_1, \textbf{j}_1),\dots,(\textbf{i}_N, \textbf{j}_N))=\sum_{n=1}^{N}\textbf{M}((\textbf{i}_n, \textbf{j}_n))
\end{equation}
where $\textbf{M}((\textbf{i},\textbf{j}))=(\textbf{f}(\textbf{i})-\textbf{f}(\textbf{j}))(\textbf{f}(\textbf{i})-\textbf{f}(\textbf{j}))^{\top}$ is the information of a single pair $(\textbf{i},\textbf{j})$.
The performance of the statistical analysis based on a paired comparison experiment depends on the pairs (alternatives) in the choice sets that are presented. The choice of such pairs $(\textbf{i}_1,\textbf{j}_1),\dots,(\textbf{i}_N,\textbf{j}_N)$ is called a design. 
\par
In the present paper we restrict our attention to approximate or continuous designs $\xi$ as detailed in \citet{kiefer1959optimum}, which are defined as discrete probability measures on the design region $\mathcal{X}$ of all pairs $(\textbf{i},\textbf{j})$. Moreover, every approximate design $\xi$ which assigns only rational weigths $\xi(\textbf{i},\textbf{j})$ to all pairs $(\textbf{i},\textbf{j})$ in its support points can be realized as an exact design $\xi_N=((\textbf{i}_1,\textbf{j}_1),\dots,(\textbf{i}_N,\textbf{j}_N))$ for some sample size $N$. 
\par
The standardized (per observation) information matrix of an approximate design $\xi$ in the linear paired comparison model \eqref{eq:2} is defined by
\begin{equation}\label{eq:4}
\textbf{M}(\xi)=\sum_{(\textbf{i}, \textbf{j})\in\mathcal{X}}\xi(\textbf{i},\textbf{j})\textbf{M}((\textbf{i},\textbf{j})).
\end{equation} 
Note that for an exact design $\xi_N=((\textbf{i}_1,\textbf{j}_1),\dots,(\textbf{i}_N,\textbf{j}_N))$ the relation $\textbf{M}(\xi_N)=\frac{1}{N}\textbf{M}((\textbf{i}_1,\textbf{j}_1),\dots,(\textbf{i}_N,\textbf{j}_N))$ holds.
\par
Optimality criteria for approximate designs $\xi$ are functionals of $\textbf{M}(\xi)$. As in a majority of works about optimal designs for paired comparison experiments, here we confine ourselves to the $D$-optimality criterion which aims at maximizing the determinant of the information matrix. For instance, an approximate design $\xi^{\ast}$ is $D$-optimal if it maximizes the determinant of the information matrix, that is, if det$\textbf{M}(\xi^{\ast})$ $\geq$ det$\textbf{M}(\xi)$ for every approximate design $\xi$.

\vspace{5mm}

\noindent
\textbf{Example 1. } One-way layout (two levels)\par
For illustrative purposes we first consider the situation of just one attribute ($K=1$) which may vary only over two levels ($i,j=1,2$), and  we adopt the standard parameterization of effect-coding \citep[see][]{grasshoff2004optimal}. 
In this setting, the effects of each single level $i=1,2$ has parameters $\alpha_{i}$ satisfying the identifiability condition $\alpha_1+\alpha_2=0$. 
Hence,
\begin{equation}\label{eq:5}
\begin{split}
\tilde{Y}_{na}(i)&=\mu_{n}+ \alpha_{i}+\tilde{\epsilon}_{na}
\end{split}
\end{equation}
with $i\in\mathcal{I}=\{1,2\}$, where $\mu_{n}$ denotes the block effects, $n=1,\dots,N$, and $\tilde{\epsilon}_{na}$ the random error is assumed to be uncorrelated with constant variance and zero mean. For effect-coding the regression function $f=g$ is given by $g(1)=1$ and $g(2) = -1$, respectively. This ensures the usual identifiability condition  $\alpha_2=-\alpha_1$ such that
\begin{equation}\label{eq:6}
\begin{split}
\tilde{Y}_{na}(i)&=\mu_{n}+ g(i)\beta+\tilde{\epsilon}_{na},
\end{split}
\end{equation}
where $\beta=\alpha_1=-\alpha_2$.
\par
Then for paired comparisons an observation of the effects $\alpha_i - \alpha_j$ of level $i$ compared to level $j$ can be characterized by the response
\begin{equation}\label{eq:7}
\begin{split}
Y_{n}(i,j)&= (g(i)-g(j))\beta+\varepsilon_n = \alpha_i-\alpha_j+\varepsilon_n.
\end{split}
\end{equation}
Note that $\textbf{M}((i,j))=(\pm 2)^2=4$ for $i \neq j$, while $\textbf{M}((i,i))=0$. From this it is obvious that only pairs with different levels should be used and that, in particular, the design $\bar{\xi}$ which assigns equal weight $1/2$ to each of the two pairs $(1,2)$ and $(2,1)$ is optimal with resulting information
\begin{equation}\label{eq:8}
\begin{split}
\textbf{M}(\bar{\xi})=4.
\end{split}
\end{equation}
This design $\bar{\xi}$ will serve as a brick for constructing optimal designs in situations with more than one attribute later on.

\section{Model with second-order interactions}
For the setting of binary attributes the results for main effects and first-order interactions have been derived by \citet{van1987optimal} and have been extended to more than two levels by \citet{grasshoff2003optimal}. 
Here, we focus on second-order interaction effects for binary attributes. 
\par
In what follows, we take into consideration $k=1,\dots,K$ attributes having two levels each that are assumed to derive the preferences for the alternatives in a paired comparison experiment. In paired comparison experiments the alternatives are represented by combinations of attribute levels. For alternatives in a choice set, we denote by $\textbf{i}=(i_1,\dots,i_K)$ the first alternative and the second alternative by $\textbf{j}=(j_1,\dots,j_K)$ which are both elements of the set $\boldsymbol{\mathcal{I}}=\{1,2\}^{K}$ where the numbers $1$ and $2$ represent the first and second level of each attribute. The choice set $(\mathbf{i},\mathbf{j})$ is an ordered pair of alternatives $\textbf{i}$ and $\textbf{j}$ which is chosen from the design region $\mathcal{X}=\mathcal{I}\times\mathcal{I}$. For each attribute (component) $k$ the corresponding marginal model coincides with that of the one-way layout with regression functions $f_k=g$ as defined in Example~1.
\par
More formally, for the case of direct response $\tilde{Y}_{na}$ at alternative $\mathbf{i}=(i_1,\ldots,i_K)$, we consider the second-order interactions model 
\begin{equation}\label{eq:9}
\begin{split}
\tilde{Y}_{na}(\mathbf{i})&=  \mu_{n}+ \sum^{K}_{k=1}\alpha^{(k)}_{i_{k}} + \sum_{k<\ell}\alpha^{(k,\ell)}_{i_{k},i_{\ell}}+ \sum_{k<\ell<m}\alpha^{(k,\ell,m)}_{i_{k},i_{\ell},i_{m}}+ \tilde{\varepsilon}_{na},
\end{split}
\end{equation}
where $\alpha_{i_k}^{(k)}$ is the main effect of the $k$-th attribute when the corresponding level is $i_k$, $\alpha_{i_k,i_\ell}^{(k,\ell)}$ is the first-order interactions effect of the $k$-th and $\ell$-th attribute when the corresponding levels are $i_k$ and $i_\ell$, respectively, and $\alpha_{i_k,i_\ell,i_m}^{(k,\ell,m)}$ is the second-order interactions effect of the $k$-th,  $\ell$-th and $m$-th attribute when the corresponding levels are $i_k$, $i_\ell$ and $i_m$, respectively.  
\par
Moreover, by the common identifiability conditions of effect-coding the following equalities hold:
$\alpha_1^{(k)}=\beta_k$ and $\alpha_2^{(k)}=-\beta_k$, $\alpha_{1,1}^{(k,\ell)}=\alpha_{2,2}^{(k,\ell)}=\beta_{k,\ell}$ and $\alpha_{1,2}^{(k,\ell)}=\alpha_{2,1}^{(k,\ell)}=-\beta_{k,\ell}$, $\alpha_{1,1,1}^{(k,\ell,m)}=\alpha_{1,2,2}^{(k,\ell,m)}=\alpha_{2,1,2}^{(k,\ell,m)}=\alpha_{2,2,1}^{(k,\ell,m)}=\beta_{k,\ell,m}$ and $\alpha_{1,1,2}^{(k,\ell,m)}=\alpha_{1,2,1}^{(k,\ell,m)}=\alpha_{2,1,1}^{(k,\ell,m)}=\alpha_{2,2,2}^{(k,\ell,m)}=-\beta_{k,\ell,m}$.
Then
$$\boldsymbol{\beta}=(\beta_1,\dots,\beta_K,\beta_{1,2},\dots,\beta_{K-1,K},\beta_{1,2,3},\dots,\beta_{K-2,K-1,K})^\top$$ is a minimal vector of parameters, where $(\beta_1,\dots,\beta_K)^\top$ is the vector of main effects of dimension $p_1=K$, $(\beta_{1,2},\dots,\beta_{K-1,K})^{\top}$ is the vector of first-order interactions of dimension $p_2={K \choose 2}$, and $(\beta_{1,2,3},\dots,\beta_{K-2,K-1,K})^\top$ is the vector of second-order interactions of dimension $p_3={K \choose 3}$.
Hence, $\boldsymbol{\beta}$ is a vector of dimension $p=p_1+p_2+p_3=K(K^2+5)/6$. 
 \par
With the above notation the model can be rewritten as

\begin{align}\label{eq:10}
\tilde{Y}_{na}(\mathbf{i})&=\mu_n+\sum_{k=1}^{K}\beta_{k}f_k(i_k)+\sum_{k<\ell}\beta_{k,\ell}f_k(i_k)f_\ell(i_\ell) \nonumber \\
&\qquad+\sum_{k<\ell<m}\beta_{k,\ell,m}f_k(i_k)f_\ell(i_\ell)f_m(i_m)+\tilde{\varepsilon}_{na},
\end{align}
with the vector
\begin{align}\label{eq:11}
\textbf{f}(\textbf{i})=(f_1(i_1),\dots,f_K(i_K),f_1(i_1)f_2(i_2),\dots,f_{K-1}(i_{K-1})f_K(i_K),\nonumber \\
\qquad f_1(i_1)f_2(i_2)f_3(i_3),\dots,f_{K-2}(i_{K-2})f_{K-1}(i_{K-1})f_K(i_K))^{\top} 
\end{align}
of regression functions of dimension $p$. 
Also here, the first $K$ entries $f_1(i_{1}),\dots,f_K(i_{K})$ of $\textbf{f}(\textbf{i})$ are associated with the main effects.
Moreover, the second set of entries $f_1(i_{1})f_2(i_{2}),\dots,f_{K-1}(i_{K-1})f_K(i_{K})$ of $\textbf{f}(\textbf{i})$ are associated with the first-order interactions and has dimension $p_2=K(K-1)/2$, and the remaining entries $f_1(i_{1})f_2(i_{2})f_3(i_{3}),\dots,f_{K-2}(i_{K-2})f_{K-1}(i_{K-1})f_K(i_{K})$ of $\textbf{f}(\textbf{i})$ are associated with the second-order interactions and has dimension $p_3=K(K-1)(K-2)/6$.
\par
Note that $f_k=g$ from Example~1 for all $k=1,\ldots,K$. 
Finally, the resulting paired comparison model is given by
\begin{align}\label{eq:12}
Y_{n}(\textbf{i},\textbf{j})&=\sum_{k=1}^{K}(g(i_k)-g(j_k))\beta_{k}+\sum_{k<\ell}(g(i_k)g(i_\ell)-g(j_k)g(j_\ell))\beta_{k,\ell}  \nonumber \\
&\qquad+\sum_{k<\ell<m}(g(i_k)g(i_\ell)g(i_m)-g(j_k)g(j_\ell)g(j_m))\beta_{k,\ell,m}+\varepsilon_n,
\end{align}
where the corresponding regression functions may only attain the values $-2$, $0$ or $+2$.
\section{Optimal designs}
We consider the second-order interactions paired comparison model \eqref{eq:12} with corresponding regression functions $\textbf{f}(\textbf{i})$ given by \eqref{eq:11}. 
For what follows, we introduce the comparison depth $d$ as in \citet{grasshoff2003optimal} which describes the number of attributes in which the two alternatives differ. 
These sets are very essential for optimal designs construction. 
The design region $\mathcal{X}$ can be partitioned into disjoint sets such that the pairs in each set differ only in a fixed number $d$ of the attributes. 
More precisely, for a comparison depth $d=0,\dots,K$, let $\mathcal{X}_{d}=\{(\textbf{i},\textbf{j})\in\mathcal{X}: |\{k: i_{k}\neq j_{k}\}|=d\}$ be the set of all pairs of alternatives which differ in exactly $d$ attributes. 
These sets constitute the orbits with respect to permutations of both the levels $i_k=1,2$ within the attributes as well as among attributes $k=1,\dots,K$ themselves.
Note that the $D$-criterion is invariant with respect to those permutations \citep[see][]{1996optimum}. 
As a result, it is sufficient to look for optimality in the class of invariant designs which are uniform on the orbits of fixed comparison depth. 
\par
Denote by $N_{d}=2^{K}{K \choose d}$ the number of different (ordered) pairs in $\mathcal{X}_{d}$ which vary in exactly $d$ attributes and by $\bar{\xi}_{d}$ the uniform approximate design which assigns equal weights $\bar{\xi}_{d}(\textbf{i},\textbf{j})=1/N_{d}$ to each pair $(\textbf{i},\textbf{j})$ in $\mathcal{X}_{d}$ and weight zero to all remaining pairs in $\mathcal{X}$.  
We next obtain the information matrix for these invariant designs.
\newtheorem{lemma}{Lemma}
\begin{lemma}\label{lemma1}
Let $d$ be a fixed comparison depth. The uniform design $\bar{\xi}_{d}$ on the set $\mathcal{X}_{d}$ of comparison depth $d$ has a diagonal information matrix
\begin{multline*}
\begin{split}
\mathbf{M}(\bar{\xi}_{d})=\begin{pmatrix} h_{1}(d)\mathbf{Id}_{K}&\mathbf{0}&\mathbf{0}\\
 \mathbf{0} & h_{2}(d)\mathbf{Id}_{K\choose2}&\mathbf{0}\\
 \mathbf{0} &\mathbf{0}&h_{3}(d)\mathbf{Id}_{K\choose3}\end{pmatrix},
 \end{split}
\end{multline*}     
\begin{equation*}
\resizebox{1.0\hsize}{!}{
where $h_{1}(d) = \frac{4d}{K}, h_{2}(d) =\frac{8d(K-d)}{K(K-1)}\ and \ h_{3}(d) =\frac{4d(3K^{2}-6dK+4d^{2}-3K+2)}{K(K-1)(K-2)}$}.
\end{equation*}   
\end{lemma}
Here, $\mathbf{Id}_m$ denotes the identity matrix of order $m$. 
The proof of Lemma~\ref{lemma1} is given in the Appendix.
\par
Note that for $d=0$ all pairs have identical attributes ($\mathbf{i}=\mathbf{j}$), $h_r(0)=0$ for $r=1,2,3$, and the information is zero. Hence, the comparison depth $d=0$ can be neglected. Moreover, every invariant design $\bar{\xi}$ can be written as a convex combination $\bar{\xi}=\sum^{K}_{d=1}w_{d}\bar{\xi}_{d}$ of uniform designs on the comparison depths $d$ with corresponding weights $w_{d}\geq 0$, $\sum^{K}_{d=1}w_{d}=1$. Consequently, every invariant design has also diagonal information matrix.
\begin{lemma}\label{lemma2}
Let $\bar{\xi}$ be an invariant design on $\mathcal{X}$, i.\,e.\ $\bar{\xi}=\sum^{K}_{d=1}w_{d}\bar{\xi}_{d}$, then $\bar{\xi}$ has diagonal information matrix
\begin{multline*}
\begin{split}
\mathbf{M}(\bar{\xi})=\begin{pmatrix} h_{1}(\bar{\xi})\mathbf{Id}_{K}&\mathbf{0}&\mathbf{0}\\
 \mathbf{0} & h_{2}(\bar{\xi})\mathbf{Id}_{K\choose2}&\mathbf{0}\\
 \mathbf{0}&\mathbf{0}&h_{3}(\bar{\xi})\mathbf{Id}_{K\choose3}\end{pmatrix},
 \end{split}
\end{multline*}
where $h_{r}(\bar{\xi})=\sum_{d=1}^Kw_dh_r(d)$, $r=1,2,3$.  
\end{lemma}
Next, we consider optimal designs for the main effects, the first-order interaction and the second-order interaction terms separately by maximizing the corresponding entries $h_{1}(d)$, $h_{2}(d)$ and $h_{3}(d)$, respectively, in the information matrix. 
The resulting designs are optimal with respect to any invariant criterion including $D_A$-optimality for the corresponding subset of the parameter vector.
\newtheorem{theorem}{Theorem}
\begin{theorem}\label{theorem1}
Let $d_1^{\ast}=K$. Then the uniform design $\bar{\xi}_{d_1^{\ast}}=\bar{\xi}_{K}$ on the largest possible comparison depth $K$ is universally optimal for the main effects $(\beta_{1},\dots,\beta_{K})^{\top}$.
\end{theorem}
This means that for the main effects only those pairs of alternatives should be used which differ in all attributes.
The proof of Theorem~\ref{theorem1} follows directly from $h_1(d)=4d/K$ which attains its maximum at $d_1^{\ast}=K$.
\begin{theorem}\label{theorem2}
Let $d^{\ast}=K/2$ for $K$ even and $d^{\ast}=(K-1)/2$ or $d^{\ast}=(K+1)/2$ for $K$ odd, respectively.
Then the uniform design $\bar{\xi}_{d_2^{\ast}}$ is universally optimal for the first-order interaction effects $(\beta_{1,2},\dots,\beta_{K-1,K})^{\top}$.
\end{theorem}
This means that for the first-order interactions only those pairs of alternatives should be used which differ in about half of the attributes.
For the proof of Theorem~\ref{theorem2} note that $h_2(0)=h_{2}(K)=0$ and $h_2$ is a quadratic polynomial in $d$ with negative leading coefficient. Then $h_2$ attains its maximum at the middle point $K/2$ which is integer for $K$ even. For $K$ odd the maximum occurs at both of the symmetrical adjacent  integers $(K-1)/2$ and $(K+1)/2$.
\par
It is worth-while mentioning that $h_1(d)$ and $h_2(d)$ are identical to the corresponding values in the first-order interactions model considered in \citet{grasshoff2003optimal} and, hence, the optimal designs of Theorem~\ref{theorem1} and \ref{theorem2} are the same as in the first-order interactions model. 
However, for the second-order interactions the following result is new.
\begin{theorem}\label{theorem3}
Let $d_3^{\ast}=1$ or $d_3^{\ast}=3$ for $K=3$ and $d_3^{\ast}=K$ for $K\geq 4$, respectively.
Then the uniform design $\bar{\xi}_{d_3^{\ast}}$ is universally optimal for the second-order interaction effects $(\beta_{1,2,3},\dots,\beta_{K-2,K-1,K})^{\top}$.
\end{theorem}
This means that also for the second-order interactions only those pairs of alternatives should be used which differ in all attributes.
\begin{proof}[Proof of Theorem~\ref{theorem3}]
For $K=3$ we get $h_3(1)=h_3(3)=4$ and $h_3(2)=0$ which establishes the result in this case.
\par
For $K\geq 4$ note that the function $h_{3}$ is a cubic polynomial in the comparison depth $d$ with positive leading coefficient.
Extended to a function on the real line $h_3$ is point symmetric with respect to $(K/2,h_3(K)/2)$ and attains its local maximum and local minimum at $d_{3, \max}=K/2-\sqrt{9K-6}/6$ and $d_{3, \min}=K/2+\sqrt{9K-6}/6$, respectively. 
Now, the numerator of $h_3(d_{3, \min})$ is proportional to $d_{3, \min}^2(3K-4d_{3, \min})$. Inserting the solution for $d_{3, \min}$ into the last factor yields $3K-4d_{3, \min}=K-2\sqrt{K-2/3}$ which is equal to $0.349$ for $K=4$ and increasing in $K\geq 3$. 
Hence, $h_3(d_{3, \min})>0$ for $K\leq 4$ and, by symmetry, $h_3(d_{3, \max})<h_3(K)$ which proves the result.
\end{proof}
It is worth-while mentioning that a single comparison depth $d$ may be sufficient for non-singularity of the information matrix $\mathbf{M}(\bar{\xi}_d)$, i.e.\ for the identifiability of all parameters. This can be easily seen by observing $h_r(1)>0$, $r=1,2,3$, for $d=1$. But this is not true for all comparison depths as $h_2(K)=0$.
Moreover, in view of Theorems \ref{theorem1}, \ref{theorem2} and \ref{theorem3} no design exists which is  universally optimal for the whole parameter vector. 
As a consequence, we confine ourselves to the $D$-criterion to derive optimal design for the whole parameter vector.
\par
Define by 
$v((\textbf{i},\textbf{j}),\xi)=(\textbf{f}(\textbf{i})-\textbf{f}(\textbf{j}))^{\top}\textbf{M}(\xi)^{-1}(\textbf{f}(\textbf{i})-\textbf{f}(\textbf{j}))$ the variance function for the design $\xi$ which plays an important role for the $D$-criterion. 
According to the Kiefer-Wolfowitz equivalence theorem \citep{kiefer1960equivalence} a design $\xi^{\ast}$ is $D$-optimal if the associated variance function is bounded by the number of parameters, $v((\mathbf{f}(\mathbf{i}),\mathbf{f}(\mathbf{j})),\xi^{\ast})\leq p$.
Now, for invariant designs $\bar{\xi}$ the variance function $v((\textbf{i},\textbf{j}),\bar{\xi})$ is also invariant with respect to permutations of levels and attributes and is, hence, constant on the orbits $\mathcal{X}_{d}$ of fixed comparison depth $d$. 
Hence, the value of the variance function for an invariant design $\bar{\xi}$ evaluated at comparison depth $d$ may be denoted by $v(d,\bar{\xi})$, say, where $v(d,\bar{\xi})=v((\mathbf{i},\mathbf{j}),\bar{\xi})$ on $\mathcal{X}_{d}$. The following results provide formulae for calculating the variance function.
\begin{theorem}\label{thrm4}
For every invariant design $\bar{\xi}$ the variance function $v(d,\bar{\xi})$ is given by 
\begin{equation*}
v(d,\bar{\xi})=4d\left(\frac{1}{h_{1}(\bar{\xi})}+\frac{K-d}{h_{2}(\bar{\xi})}+\frac{3K^{2}-6dK+4d^{2}-3K+2}{6h_{3}(\bar{\xi})}\right)
\end{equation*}
\end{theorem}
The proof is given in the Appendix.
If the invariant design $\bar{\xi}$ is concentrated on a single comparison depth, then this representation simplifies.
\newtheorem{corollary}{Corollary}
\begin{corollary}\label{cor_thrm4}
For a uniform design $\bar{\xi}_{d^{\prime}}$ on a single comparison depth $d^{\prime}$ the variance function is given by  \begin{equation*}
v(d,\bar{\xi}_{d^{\prime}})=\frac{d}{d^{\prime}}\left(p_{1}+p_{2}\frac{K-d}{K-d^{\prime}}+p_{3}\frac{3K^{2}-6dK+4d^{2}-3K+2}{3K^{2}-6d^{\prime}K+4d^{\prime2}-3K+2}\right).
\end{equation*}
\end{corollary}
As a consequence, for $d=d^{\prime}$, we immediately obtain $v(d,\bar{\xi}_{d})=p_1+p_2+p_3=p$ which recovers the $D$-optimality of $\bar{\xi}_d$ on $\mathcal{X}_d$ in view of the Kiefer-Wolfowitz equivalence theorem. 
The following result gives an upper bound on the number of comparison depths required for a $D$-optimal design.
\begin{theorem}\label{theorem5}
In the second-order interactions model the $D$-optimal design $\xi^{\ast}$ is supported on, at most, three different comparison depths $K$, $d^{\ast}$ and $d^{\ast}+1$, say, i.e.\  
$\xi^{\ast}=w^{\ast}_{K}\xi_K+w^{\ast}_{d^{\ast}}\xi_{d^{\ast}}+(1-w^{\ast}_{K}-w^{\ast}_{d^{\ast}})\xi_{d^{\ast}+1}$.
\end{theorem}
\begin{proof}
 According to a corollary of the Kiefer-Wolfowitz equivalence theorem for the $D$-optimal design $\xi^{\ast}$ the variance function $v(d,\xi^{\ast})$ is equal to the number of parameters $p$ for all $d$ such that $w_d^{\ast}>0$. By Theorem~\ref{thrm4} the variance function is a cubic polynomial in the comparison depth $d$ with positive leading coefficient. According to the fundamental theorem of algebra the variance function $v(d,\xi^{\ast})$ may thus be equal to $p$ for, at most, three different values $d_1<d_2<d_3$ of $d$, say. Now, by the Kiefer-Wolfowitz equivalence theorem itself $v(d,\xi^{\ast})\leq p$ for all $d=0,1,\dots,K$. Hence, by the shape of the variance function we obtain that in the case of three different comparison depths $d_3=K$ and $d_2=d_1$ must hold. For two comparison depths either $d_2=K$ or two adjacent comparison depths $d_1$ and $d_2=d_1+1$ are included.
\end{proof}
For $K=3$ the $D$-optimal design can be given explicitly.
\begin{theorem}\label{theorem6}
For $K=3$ the  uniform design $\xi^{\ast}=\frac37\bar{\xi}_1+\frac37\bar{\xi}_2+\frac17\bar{\xi}_3$ on all pairs with non-zero comparison depth is $D$-optimal in the second-order interactions model.
\end{theorem}
\begin{proof}
For $K=3$ the second-order interactions model is a full interactions model. Hence, the results follows directly from Theorem~4 in \citet{grasshoff2003optimal}.
\par
Alternatively, we may compute the variance function by first computing $h_r(\xi^{\ast})=16/7$, $r=1,2,3$, and then deriving $v(d,\xi^{\ast})=7d(d^2-6d+11)/6$ which results in $v(d,\xi^{\ast})=7$ for $d=1,2,3$. Because $p=7$ for $K=3$ the $D$-optimality of $\xi^{\ast}$ follows from the Kiefer-Wolfowitz equivalence theorem.
\end{proof}
Hence, for $K=3$ all three comparison depths are needed for $D$-optimality.
For $K\geq 4$ numerical computations indicate that only two different comparison depths $K$ and $d^{\ast}$ are required.
In Table~\ref{tab:1} numerical solutions are presented for numbers $K$ of attributes between $4$ and $10$.
\begin{table}[H]
\caption{Optimal comparison depths and optimal weights for $K$ binary attributes}\label{tab:1} 
\begin{center}
\begin{tabular}{c|c*{7}{c}}
$K$
&\makebox[1.5em]{4}&\makebox[1.5em]{5}&\makebox[1.5em]{6}&\makebox[1.5em]{7}&\makebox[1.5em]{8}&\makebox[1.5em]{9}&\makebox[1.5em]{10}\\
$w_{K}^{\ast}$
&0.143&0.167&0.268&0.303&0.356&0.423&0.462\\\hline
$d^{\ast}$
&2&2&3&3&3&4&4\\
$w_{d^*}^{\ast}$
&0.857&0.833&0.732&0.697&0.644&0.577&0.538
\end{tabular}
\end{center}
\end{table} 
For fixed number $K$ of attributes and intermediate comparison depth $d$ the optimal weights  $w_K^{\ast}$ and  $w_d^{\ast}=1-w_K^{\ast}$ have been determined analytically by direct maximization of $\ln(\det(\mathbf{M}(w_K\bar{\xi}_K+(1-w_K)\bar{\xi}_d)))$.
In particular, for $K=5$, $7$ and $9$ we obtain the optimal intermediate comparison depth $d^{\ast}=(K-1)/2$ and corresponding optimal weights 
\begin{equation*}
\resizebox{1.0\hsize}{!}{
$w_K^*=\frac{2K^3-6K^2+7K-K\sqrt{K^{6}-12K^{5}+64K^{4}-198K^{3}+448K^{2}-636K+369}\ +15}{-K^{4}+2K^{3}-2K^{2}+10K+15}$}
\end{equation*}
and $w_{d^{\ast}}^{\ast}=1-w_K^{\ast}$.
For $K=4$ and $6$ we get the optimal intermediate comparison depth $d^{\ast}=K/2$ and corresponding optimal weights 
\begin{equation*}
w_K^*=\frac{K^2-6K+11}{K^2+5}
\end{equation*}
and $w_{d^{\ast}}^{\ast}=1-w_K^{\ast}$, respectively,
while for $K=8$ and $10$ the optimal intermediate comparison depth is $d^{\ast}=(K/2)-1$ with corresponding optimal weights 
\begin{equation*}
\resizebox{0.8\hsize}{!}{
$w_K^*=\frac{K^3+5K+(K-K^2)\sqrt{K^{4}-10K^{3}+37K^{2}-60K+180}\ +30}{-K^{4}+K^{3}+K^{2}+5K+30}$},
\end{equation*}
and $w_{d^{\ast}}^{\ast}=1-w_K^{\ast}$, respectively. The analytical results were computed using the Maple (version 2017) software (\citeauthor{maple}, 1981-2017).
The optimality of these designs has been checked numerically by virtue of the Kiefer-Wolfowitz equivalence theorem. The corresponding values of the normalized variance function $v(d,\xi^{\ast})/p$ are recorded in Table~\ref{tab:2} in the Appendix, where maximal values less or equal to $1$ establish optimality. 
\par
Based on numerical findings for larger $K$ up to $100$ we conjecture that for all numbers $K\geq 4$ of attributes only one intermediate comparison depth $d^{\ast}$ is required for $D$-optimality which is equal to $(K-1)/2$ for $K$ odd and $(K/2)-1$ for $K\geq 8$ even, respectively.
\section{Discussion}
For paired comparisons in an additive (main effects only) model optimal designs require that the components of the alternatives in the choice sets show distinct levels in all attributes \citep[see][]{grasshoff2004optimal}. 
In a first-order interactions model we have to consider pairs for an optimal design in which approximately one half of the attributes are distinct and one half of the attributes coincide \citep[see][]{grasshoff2003optimal}.
Here, it is shown that in a second-order interactions model we have to consider both types of pairs in which either all attributes have distinct levels or approximately one half of the attributes are distinct and one half of the attributes coincide to obtain a $D$-optimal design for the whole parameter vector. 
The invariance considerations used here can be readily extended to partial profiles, to higher-order interaction models and to larger numbers of levels for each attribute.
\vspace{5mm} \\
\textbf{Acknowledgement.}
This work was partially supported by Grant - Doctoral Programmes in Germany, 2016/2017 (57214224) - of the German Academic Exchange Service (DAAD).
\bibliographystyle{apa}     
\bibliography{reference3rd}   

\begin{thebibliography}{}

\bibitem[\protect\astroncite{Bradley and Terry}{1952}]{bradley1952rank}
Bradley, R.~A. and Terry, M.~E. (1952).
\newblock Rank analysis of incomplete block designs: I. {T}he method of paired
  comparisons.
\newblock {\em Biometrika}, 39:324--345.

\bibitem[\protect\astroncite{Gra{\ss}hoff et~al.}{2003}]{grasshoff2003optimal}
Gra{\ss}hoff, U., Gro{\ss}mann, H., Holling, H., and Schwabe, R. (2003).
\newblock Optimal paired comparison designs for first-order interactions.
\newblock {\em Statistics}, 37:373--386.

\bibitem[\protect\astroncite{Gra{\ss}hoff et~al.}{2004}]{grasshoff2004optimal}
Gra{\ss}hoff, U., Gro{\ss}mann, H., Holling, H., and Schwabe, R. (2004).
\newblock Optimal designs for main effects in linear paired comparison models.
\newblock {\em Journal of Statistical Planning and Inference}, 126:361--376.

\bibitem[\protect\astroncite{Gro{\ss}mann et~al.}{2002}]{grossmann2002advances}
Gro{\ss}mann, H., Holling, H., and Schwabe, R. (2002).
\newblock Advances in optimum experimental design for conjoint analysis and
  discrete choice models.
\newblock In {\em \rm{{F}ranses, {P}. {H}., {M}ontgomery, {A}. {L}. ({E}ds.)},
  \it{{A}dvances in {E}conometrics, Vol. 16. {E}conometric {M}odels in
  {M}arketing}}. JAI Press, Amsterdam, pp. 93--117.

\bibitem[\protect\astroncite{Gro{\ss}mann and
  Schwabe}{2015}]{grossmann2015handbook}
Gro{\ss}mann, H. and Schwabe, R. (2015).
\newblock Design for discrete choice experiments.
\newblock In {\em \rm{Dean, {A}., {M}orris, {M}., {S}tufken, {J}., {B}ingham,
  {D}. ({E}ds.)}, \it{{H}andbook of {D}esign and {A}nalysis of {E}xperiments}}.
  Chapman and Hall/CRC, Boca Raton, FL, pp. 787--831.

\bibitem[\protect\astroncite{Kiefer}{1959}]{kiefer1959optimum}
Kiefer, J. (1959).
\newblock Optimum experimental designs.
\newblock {\em Journal of the Royal Statistical Society. Series B},
  21:272--304.

\bibitem[\protect\astroncite{Kiefer and
  Wolfowitz}{1960}]{kiefer1960equivalence}
Kiefer, J. and Wolfowitz, J. (1960).
\newblock The equivalence of two extremum problems.
\newblock {\em Canadian Journal of Mathematics}, 12:363--366.

\bibitem[\protect\astroncite{Louviere et~al.}{2000}]{louviere2000stated}
Louviere, J.~J., Hensher, D.~A., and Swait, J.~D. (2000).
\newblock {\em Stated {C}hoice {M}ethods: {A}nalysis and {A}pplication}.
\newblock Cambridge University Press, Cambridge.

\bibitem[\protect\astroncite{{Maple Inc.}}{2017}]{maple}
{Maple Inc.} (1981-2017).
\newblock Maplesoft, a division of {W}aterloo {M}aple {I}nc., {W}aterloo,
  {O}ntario.

\bibitem[\protect\astroncite{Schwabe}{1996}]{1996optimum}
Schwabe, R. (1996).
\newblock Optimum designs for multi-factor models.
\newblock {\em Lecture Notes in Statistics 113: {S}pringer, {N}ew {Y}ork}.

\bibitem[\protect\astroncite{Schwabe et~al.}{2003}]{schwabe2003optimal}
Schwabe, R., Gra{\ss}hoff, U., Gro{\ss}mann, H., and Holling, H. (2003).
\newblock Optimal $2^{K}$ paired comparison designs for partial profiles.
\newblock In {\em PROBA\-STAT\ 2002, Proceedings of the Fourth International
  Conference on Mathematical Statistics, Smolenice 2002, Tatra Mountains
  Mathematical Publications, \rm{26:79--86}}.

\bibitem[\protect\astroncite{Street and Burgess}{2004}]{street2004optimal}
Street, D.~J. and Burgess, L. (2004).
\newblock Optimal and near-optimal pairs for the estimation of effects in
  2-level choice experiments.
\newblock {\em Journal of Statistical Planning and Inference}, 118:185--199.

\bibitem[\protect\astroncite{Street and Burgess}{2007}]{street2007construction}
Street, D.~J. and Burgess, L. (2007).
\newblock {\em The {C}onstruction of {O}ptimal {S}tated {C}hoice {E}xperiments:
  {T}heory and {M}ethods}.
\newblock Wiley, Hoboken, NJ.

\bibitem[\protect\astroncite{Train}{2003}]{train2003discrete}
Train, K.~E. (2003).
\newblock {\em Discrete {C}hoice {M}ethods with {S}imulation}.
\newblock Cambridge University Press, Cambridge.

\bibitem[\protect\astroncite{van Berkum}{1987a}]{van1987optimal1}
van Berkum, E. (1987a).
\newblock Optimal paired comparison designs for factorial and quadratic models.
\newblock {\em Journal of Statistical Planning and Inference}, 15:265.

\bibitem[\protect\astroncite{van Berkum}{1987b}]{van1987optimal}
van Berkum, E. E.~M. (1987b).
\newblock {\em Optimal {P}aired {C}omparison {D}esigns for {F}actorial
  {E}xperiments}.
\newblock CWI Tract 31, Amsterdam.

\end{thebibliography}
\begin{appendices}
\section*{APPENDIX}
\begin{proof}[Proof of Lemma~\ref{lemma1}]
First we note that for the function $g$ defined in Example~1 we have
$g(i)^{2}=1$, $g(i)g(j)=-1$, $(g(i)-g(j))^{2}=4$ for $i\neq j$.
Given a fixed comparison depth $d$ we obtain for the regression functions $f_{k}=g$ associated with the $k$-th main effect 
\begin{equation*}
\begin{split}
&\sum_{(\textbf{i},\textbf{j})\in\mathcal{X}_d}(f_{k}(i_{k})-f_{k}(j_{k}))^{2}=\left(\begin{smallmatrix}K-1 \\ d-1\end{smallmatrix}\right)2^{K}\cdot 4=\left(\begin{smallmatrix}K-1 \\ d-1\end{smallmatrix}\right)2^{K+2}
\end{split}
\end{equation*}
because there are ${K-1}\choose{d-1}$$2^{K}$ pairs in $\mathcal{X}_d$ for which $i_k$ and $j_k$ differ.
Since the number $N_d$ of paired comparisons in $\mathcal{X}_{d}$ equals $N_{d}=\left(\begin{smallmatrix}K \\ d\end{smallmatrix}\right)2^{K}$, the corresponding diagonal entries $h_{1}(d)$ in the information matrix are given by 
 \begin{equation}\label{eqn:13}
\begin{split}
h_{1}(d)=&\frac{1}{N_d}\sum_{(\textbf{i},\textbf{j})\in\mathcal{X}_d}(f_{k}(i_{k})-f_{k}(j_{k}))^{2}=\frac{4d}{K}
\end{split}
\end{equation}
\cite[compare][where a slightly different normalization is used]{grasshoff2003optimal}. 
\par
For first-order interactions, we consider attributes $k$ and $\ell$, say, and distinguish between pairs in which both attributes are distinct and pairs in which only one of these attributes has distinct levels in the alternatives while the same level is presented in both alternatives for the other attribute.
\par
 In the case $i_k \neq j_k$ and $i_\ell \neq j_\ell$ we have $g(i_k)g(i_\ell)=g(j_k)g(j_\ell)$, while for $i_\ell=j_\ell$ we get $g(i_k)g(i_\ell)=-g(j_k)g(j_\ell)$. Hence
\begin{equation*}
\begin{split}
&(g(i_{k})g(i_{\ell})-g(j_{k})g(j_{\ell}))^{2}=0 \quad\mathrm{ for }\quad i_{k}\neq j_{k} \quad\mathrm{and}\quad i_{\ell}\neq j_{\ell}
\end{split}
\end{equation*}
and
\begin{equation*}
\begin{split}
&(g(i_{k})g(i_{\ell})-g(j_{k})g(j_{\ell}))^{2}=4 \quad\mathrm{ for }\quad i_{k}\neq j_{k} \quad\mathrm{and}\quad i_{\ell}=j_{\ell},
\end{split}
\end{equation*}
respectively, where (in the latter case) the roles of the attributes $k$ and $\ell$ may be interchanged. 
\par
For given attributes $k$ and $\ell$ pairs with distinct levels in both attributes occur $\left(\begin{smallmatrix}K-2 \\ d-2\end{smallmatrix}\right)2^{K}$ times in $\mathcal{X}_d$, while those which differ only in one attribute occur $\left(\begin{smallmatrix}2 \\ 1\end{smallmatrix}\right)\left(\begin{smallmatrix}K-2 \\ d-1\end{smallmatrix}\right)2^{K}$ times. As a result, for the first-order interactions the diagonal elements $h_{2}(d)$ in the information matrix are given by 
\begin{equation}\label{eqn13}
\begin{split}
h_2(d)=&\frac{1}{N_d}2\begin{pmatrix}K-2 \\ d-1\end{pmatrix}2^{K}\cdot 4=\frac{8d(K-d)}{K(K-1)}
\end{split}
\end{equation}
\cite[compare][]{grasshoff2003optimal}.
\par
Accordingly, for second-order interactions, we consider attributes $k$, $\ell$ and $m$, say, and distinguish between pairs in which all three attributes are distinct,  pairs in which two of these attributes $k$ and $\ell$, say, have distinct levels in the alternatives while the same level is presented in both alternatives for the remaining attribute and, finally, pairs in which only one of the attributes, say, $k$ has distinct levels in the alternatives while the same level is presented in both alternatives for the two remaining attributes.
Then $g(i_{k})g(i_{\ell})g(i_{m})=-g(j_{k})g(j_{\ell})g(j_{m})$ in the first and third case, while $g(i_{k})g(i_{\ell})g(i_{m})=g(j_{k})g(j_{\ell})g(j_{m})$ in the second case. Hence,
\begin{equation*}
\begin{split}
&(g(i_{k})g(i_{\ell})g(i_{m})-g(j_{k})g(j_{\ell})g(j_{m}))^{2}=4 \quad\mathrm{ for }\quad i_{k}\neq j_{k}, \ i_{\ell}\neq j_{\ell} \quad\mathrm{and}\quad i_{m}\neq j_{m},
\end{split}
\end{equation*}
\begin{equation*}
\begin{split}
&(g(i_{k})g(i_{\ell})g(i_{m})-g(j_{k})g(j_{\ell})g(j_{m}))^{2}=0 \quad\mathrm{ for }\quad i_{k}\neq j_{k}, \ i_{\ell}\neq j_{\ell} \quad\mathrm{and}\quad i_{m}=j_{m}
\end{split}
\end{equation*}
and
\begin{equation*}
\begin{split}
&(g(i_{k})g(i_{\ell})g(i_{m})-g(j_{k})g(j_{\ell})g(j_{m}))^{2}=4 \quad\mathrm{ for }\quad i_{k}\neq j_{k}, \ i_{\ell}=j_{\ell} \quad\mathrm{and}\quad i_{m}=j_{m},
\end{split}
\end{equation*}
respectively, where again the roles of the attributes $k$, $\ell$ and $m$ may be interchanged. 
\par
For given attributes $k$, $\ell$ and $m$ the pairs with distinct levels in the three attributes occur $\left(\begin{smallmatrix}K-3 \\ d-3\end{smallmatrix}\right)2^{K}$ times in $\mathcal{X}_d$ while those which differ in two attributes occur $\left(\begin{smallmatrix}3 \\ 2\end{smallmatrix}\right)\left(\begin{smallmatrix}K-3 \\ d-2\end{smallmatrix}\right)2^{K}$ times in $\mathcal{X}_d$, and those which differ only in one attribute occur $\left(\begin{smallmatrix}3 \\ 1\end{smallmatrix}\right)\left(\begin{smallmatrix}K-3 \\ d-1\end{smallmatrix}\right)2^{K}$ times. As a result, for the second-order interactions the diagonal elements $h_{3}(d)$ in the information matrix are given by
\begin{align}\label{eqn:15}
h_{3}(d)&=\frac{1}{N_d}\begin{pmatrix}\begin{pmatrix}K-3 \\ d-3\end{pmatrix}2^{K}\cdot 4+3\begin{pmatrix}K-3 \\ d-1\end{pmatrix}2^{K}\cdot 4\end{pmatrix} \nonumber\\
&=\frac{4d((d-1)(d-2)+3(K-d)(K-d-1))}{K(K-1)(K-2)} \nonumber\\
&=\frac{4d(3K^{2}-6dK+4d^{2}-3K+2)}{K(K-1)(K-2)}.
\end{align}
Finally, it can be noted that all off-diagonal entries in the information matrix vanish because the terms in the corresponding sums add up to zero due to the effect-type coding.
\end{proof}

\begin{proof}[Proof of Theorem~\ref{thrm4}]
First we note that
\begin{equation*}
\begin{split}
\mathbf{M}(\bar{\xi})^{-1}=\begin{pmatrix} \frac{1}{h_{1}(\bar{\xi})}\textbf{Id}_{K}&\mathbf{0}&\mathbf{0}\\
 \mathbf{0} & \frac{1}{h_{2}(\bar{\xi})}\mathbf{Id}_{K\choose2}&\mathbf{0}\\
 \mathbf{0}&\mathbf{0}&\frac{1}{h_{3}(\bar{\xi})}\mathbf{Id}_{K\choose3}\end{pmatrix},
 \end{split}
\end{equation*}
for the inverse of the information matrix of the design $\bar{\xi}$.
Hence, we obtain for the variance function
\begin{align}\label{eqn:16}
v((\mathbf{i},\mathbf{j}),\bar{\xi})=&(\mathbf{f}(\mathbf{i})-\mathbf{f}(\mathbf{j}))^{\top}\mathbf{M}(\bar{\xi})^{-1}(\mathbf{f}(\mathbf{i})-\mathbf{f}(\mathbf{j}))\nonumber\\
=&\frac{1}{h_1(\bar{\xi})}\sum_{k=1}^{K}(g(i_k)-g(j_k))^2 \nonumber\\
&\mbox{}+\frac{1}{h_2(\bar{\xi})}\sum_{k<\ell}(g(i_k)g(i_{\ell})-g(j_k)g(j_{\ell}))^2  \nonumber \\
&\mbox{}+\frac{1}{h_3(\bar{\xi})}\sum_{k<\ell<m}(g(i_k)g(i_{\ell})g(i_m)-g(j_k)g(j_{\ell})g(j_m))^2.
\end{align} 
\par
As in the proof of Lemma~\ref{lemma1} we note first that for the terms associated with the main effects we have $(g(i_k)-g(j_k))^{2}=4$, when $i_k \neq j_k$, and $(g(i_k)-g(j_k))^{2}=0$ otherwise.
For a pair $(\mathbf{i},\mathbf{j})$ of comparison depth $d$ there are exactly $d$ attributes for which $i_k$ and $j_k$ differ.
Hence, the first sum on the right hand side of \eqref{eqn:16} equals $4d$.
\par
Second, for the terms associated with the first-order interactions we have $(g(i_k)g(i_{\ell})-g(j_k)g(j_{\ell}))^2=4$, if either $i_k \neq j_k$ and $i_\ell=j_{\ell}$ or $i_k = j_k$ and $i_\ell\neq j_{\ell}$, and $(g(i_k)g(i_{\ell})-g(j_k)g(j_{\ell}))^2=0$ otherwise.
For a pair $(\mathbf{i},\mathbf{j})$ of comparison depth $d$ there are $d(K-d)$ first-order interaction terms for which $(i_k,i_{\ell})$ and $(j_k,j_{\ell})$ differ in exactly one attribute $k$ or $\ell$.
Hence, the second sum on the right hand side of \eqref{eqn:16} equals $4d(K-d)$.
\par
Finally, for the terms associated with the second-order interactions we have $(g(i_k)g(i_{\ell})g(i_m)-g(j_k)g(j_{\ell})g(j_m))^2=4$, if $(i_k,i_{\ell},i_m)$ and $(j_k,j_{\ell},j_m)$ differ either in all three attributes $k$, $\ell$ and $m$ or in exactly one of these attributes, and $(g(i_k)g(i_{\ell})g(i_m)-g(j_k)g(j_{\ell})g(j_m))^2=0$ otherwise.
For a pair $(\mathbf{i}, \mathbf{j})$ of comparison depth $d$ there are $\left(\begin{smallmatrix} d \\ 3\end{smallmatrix}\right)$ second-order interaction terms for which all three associated attributes differ, and there are $d\left(\begin{smallmatrix} K-d \\ 2\end{smallmatrix}\right)$ second-order interaction terms for which $(i_k,i_{\ell},i_m)$ and $(j_k,j_{\ell},j_m)$ differ in exactly one attribute.
Hence, there are
\begin{equation*}
\begin{split}
\left(\begin{smallmatrix} d \\ 3 \end{smallmatrix}\right)+d\left(\begin{smallmatrix} K-d \\ 2\end{smallmatrix}\right) &=d(d-1)(d-2)/6+d(K-d)(K-d-1)/2\\
&=d(3K^{2}-6dK+4d^{2}-3K+2)/6
\end{split}
\end{equation*} 
non-zero entries in the third sum on the right hand side of \eqref{eqn:16}, and this sum equals $4d(3K^{2}-6dK+4d^{2}-3K+2)/6$.
\par
By inserting these results into \eqref{eqn:16} we see that the value of the variance function  depends on the pair $(\mathbf{i},\mathbf{j})$ only through its comparison depth $d$ and obtain the formula proposed.
\end{proof}
\begin{proof}[Proof of Corollary~\ref{cor_thrm4}]
This representation of the variance function follows immediately by inserting the values of $h_{r}(\bar{\xi}_d)$ from Lemma \ref{lemma1} and $p_r={K \choose r}$, $r=1,2,3$, into the formula of Theorem~\ref{thrm4}.
\end{proof}
\noindent
\begin{table}[H]
\caption{Values of the variance function of the $D$-optimal design $\xi^{\ast}$ for $K$ binary attributes (boldface \textbf{1} corresponds to the optimal comparison depths)}\label{tab:2} 
\begin{tabular}{r|cc*{10}{c}}
$K$&\makebox[1.5em]{1}&\makebox[1.5em]{2}&\makebox[1.5em]{3}&\makebox[1.5em]{4}&\makebox[1.5em]{5}&\makebox[1.5em]{6}&\makebox[1.5em]{7}&\makebox[1.5em]{8}&\makebox[1.5em]{9}&\makebox[1.5em]{10}\\\hline
4 &0.875& \ $\textbf{1}$&  0.875  &$\textbf{1}$&  &&&&&\\       
5 &0.760& \ $\textbf{1}$&0.960&0.880&$\textbf{1}$&&&&&\\
6 & 0.701&0.983&$\textbf{1}$&0.906&0.855&$\textbf{1}$&&&&\\
7 & 0.615&0.917&$\textbf{1}$&0.956&0.879&0.863&$\textbf{1}$&&&\\
8 &0.559&0.872&\textbf{1}&1&0.945&0.884&0.882&$\textbf{1}$&&\\
9 &0.504&0.811&0.962&$\textbf{1}$&0.969&0.910&0.868&0.883&$\textbf{1}$&\\
10&0.462&0.763&0.932&$\textbf{1}$&0.997&0.956&0.905&0.874&0.896&$\textbf{1}$\\\hline
\end{tabular}\\
\end{table} 
\end{appendices}
\end{document}